\documentclass[journal,twocolumn]{IEEEtran}
\usepackage{amsmath, amssymb, amsthm, bm}
\usepackage{amssymb}
\usepackage{graphicx}
\usepackage{color, colortbl} 
\usepackage[dvipsnames,svgnames,x11names]{xcolor}
\makeatother
\usepackage{tabularx, boldline, multirow}
\usepackage[boxed,ruled,vlined]{algorithm2e}
\usepackage{enumerate}
\usepackage[caption=false,font=footnotesize]{subfig}
\usepackage[bookmarks=false]{hyperref}
\usepackage{cite}
\usepackage{url,comment}
\usepackage[normalem]{ulem}
\usepackage[shortlabels]{enumitem}
\usepackage{pgfplots}
\pgfplotsset{compat=1.17}
\usetikzlibrary{plotmarks}
\usepackage{tikz}
\usetikzlibrary {arrows.meta} 
\usepackage{pgfgantt}
\usepackage{mathtools}
\usepackage{accents}
\usepackage{balance}
\usepackage{acro}
\usepackage{arydshln}
\usepackage{scalerel}

\newtheorem{theorem}{Theorem}[]
\newtheorem{lemma}[theorem]{Lemma}
\newtheorem{proposition}{Proposition}

\usepackage[left=0.58 in,right=0.60 in,top=0.7 in, bottom=1 in]{geometry}
\DeclareAcronym{los}{
short=LoS,
long= line-of-sight,
}
\DeclareAcronym{mmimo}{
short=mMIMO,
long= massive MIMO,
}
\DeclareAcronym{mimo}{
short=MIMO,
long= multiple-input multiple-output,
}
\DeclareAcronym{ap}{
short=AP,
long= access point,
}
\DeclareAcronym{5g}{
short=5G,
long= fifth generation network,
}
\DeclareAcronym{ofdm}{
short=OFDM,
long= orthogonal frequency division multiplexing,
}
\DeclareAcronym{cpu}{
short=CPU,
long= central processing unit,
}
\DeclareAcronym{pn}{
short=PN,
long= phase noise,
}
\DeclareAcronym{ue}{
short=UE,
long= user equipment,
}
\DeclareAcronym{lo}{
short=LO,
long= local oscillator,
}
\DeclareAcronym{se}{
short=SE,
long= spectral efficiency,
}
\DeclareAcronym{sinr}{
short=SINR,
long= signal-to-interference-and-noise ratio,
}
\DeclareAcronym{uatf}{
short=UatF,
long= use-and-then-forget,
}
\DeclareAcronym{lmmse}{
short=LMMSE,
long= linear minimum mean square error,
}
\DeclareAcronym{mmse}{
short=MMSE,
long= minimum mean square error,
}
\DeclareAcronym{hwi}{
short=HWI,
long= hardware impairment,
}
\DeclareAcronym{mr}{
short=MR,
long= maximum ratio,
}
\DeclareAcronym{pmmse}{
short=PMMSE,
long= partial MMSE,
}
\DeclareAcronym{ha}{
short=HA,
long= hardware-aware,
}
\DeclareAcronym{fir}{
short=FIR,
long= finite impulse response,
}
\DeclareAcronym{dft}{
short=DFT,
long= discrete Fourier transform,
}
\DeclareAcronym{idft}{
short=IDFT,
long= inverse discrete Fourier transform,
}
\DeclareAcronym{rb}{
short=RB,
long= resource block,
}
\DeclareAcronym{isi}{
short=ISI,
long= intersymbol interference,
}
\DeclareAcronym{ici}{
short=ICI,
long= intercarrier interference,
}
\DeclareAcronym{cp}{
short=CP,
long= cyclic prefix,
}
\DeclareAcronym{cpe}{
short=CPE,
long= common phase error,
}
\DeclareAcronym{hu}{
short=HU,
long= hardware-unaware,
}
\DeclareAcronym{pna}{
short=PN-aware,
long= PN-aware,
}
\DeclareAcronym{pnu}{
short=PN-unaware,
long= PN-unaware,
}
\DeclareAcronym{csi}{
short=CSI,
long= channel state information,
}
\DeclareAcronym{lpmmse}{
short=LP-MMSE,
long= local-partial MMSE,
}
\DeclareAcronym{dcc}{
short=DCC,
long= dynamic cooperation clustering,
}

\begin{document}

\title{Impact of Phase Noise on\\ Uplink Cell-Free Massive MIMO OFDM\\
\thanks{This work was supported by the Swedish Foundation for Strategic Research (SSF), grant no. ID19-0021, and the Gigahertz-ChaseOn Bridge Center at Chalmers in a project financed by Chalmers, Ericsson, and Qamcom. L. Sanguinetti was partially supported by the Italian Ministry of Education and Research (MUR) in the framework of the FoReLab project (Departments of Excellence), and by the Università di Pisa under the "PRA – Progetti di Ricerca di Ateneo" (Institutional Research Grants) - Project no. PRA 2022-2023-91.
  \\
        All the simulation results can be reproduced using the Matlab code that will be available upon acceptance.}
}
\author{Yibo~Wu\IEEEauthorrefmark{1}\IEEEauthorrefmark{2},
Luca~Sanguinetti\IEEEauthorrefmark{3},
    Ulf~Gustavsson\IEEEauthorrefmark{1},
Alexandre~Graell~i~Amat\IEEEauthorrefmark{2}, and
    Henk~Wymeersch\IEEEauthorrefmark{2}\\
        \IEEEauthorrefmark{1}Ericsson Research, Gothenburg, Sweden\\
        \IEEEauthorrefmark{2}Department of Electrical Engineeering, Chalmers University of Technology, Gothenburg, Sweden\\
		\IEEEauthorrefmark{3}Dipartimento di Ingegneria dell’Informazione, University of Pisa, Pisa, Italy
        }
\maketitle

\begin{abstract}
Cell-Free \acl{mmimo} networks provide huge power gains and resolve inter-cell interference by coherent processing over a massive number of distributed instead of co-located antennas in \acfp{ap}. Cost-efficient hardware is preferred but imperfect \aclp{lo} in both \acp{ap} and users introduce multiplicative \acf{pn}, which affects the phase coherence between \acp{ap} and users even with centralized processing. In this paper, we first formulate the system model of a \ac{pn}-impaired uplink Cell-Free \acl{mmimo} \acl{ofdm} network, and then propose a \acl{pna} \acl{lmmse} channel estimator and derive a \ac{pn}-impaired uplink \acl{se} expression. Numerical results are used to quantify the \acl{se} gain of the proposed channel estimator over alternative schemes for different receiving combiners.
\end{abstract}

\section{Introduction}
\Acl{mmimo} offers phenomenal received power gains by coherently transmitting a signal over multiple antennas without increasing the transmit power~\cite{bjornson2016massive}. This coherent transmission can be implemented in two ways, classified by the deployment of antennas: deploying  co-located antennas leads to cellular  network~\cite{bjornson2017massive}, while deploying distributed antennas over \acp{ap} leads to  cell-free  network~\cite{nayebi2017precoding,ngo2017cell,demir2021foundations}. 
Coherent transmissions in  cell-free networks relies on both time and phase synchronization among \acp{ap}~\cite{zheng2023asynchronous}. Even if  there is a centralized \ac{cpu} that synchronizes all \acp{ap}, the imperfect \acp{lo} at both \acp{ap} and \acp{ue} introduce different \ac{pn} that varies over time, which unavoidably affects the transmission coherence and reduces the  power gain~\cite{bjornson2015massive}. Using better quality \acp{lo} can alleviate the \ac{pn} problem, while the corresponding hardware cost scales up with up to hundreds of \acp{lo} in \acp{ap} and \acp{ue}. Thus, to design an economical and reliable cell-free  network, it is important to evaluate the relation of the \ac{pn} impact and the \ac{lo} quality.

There exist several works investing the \ac{pn} impact on either cellular or cell-free  massive MIMO networks~\cite{bjornson2015massive,papazafeiropoulos2021scalable,zheng2023asynchronous,jin2020spectral,ozdogan2019performance,pitarokoilis2016performance}. However, the vast majority of these works  model the \ac{pn} in a single-carrier fashion, while most modern communication systems utilize \ac{ofdm}. The loss of orthogonality between \ac{ofdm} subcarriers in the presence of \ac{pn} is ignored by the conventional single-carrier models. This may have an impact on the design of the network. For example, it may result into a mismatched channel estimator and a mismatched combining/precoding scheme. The authors in~\cite{pitarokoilis2016performance} studied the \ac{pn} impact on the uplink achievable \acf{se} in an \ac{ofdm} system under the assumption of perfect \acl{csi}, while a practical \acl{pna} channel estimation and the corresponding achievable \ac{se} are not studied.

The aim of this paper is to evaluate the impact of \ac{pn} in the uplink of cell-free \ac{ofdm} \acl{mmimo}. To this end, we first provide the \ac{pn}-impaired system model, which correctly models the impact of the time-domain \ac{pn} on any subcarrier of the received frequency-domain OFDM symbols. The model is then used to derive the \acl{pna} \ac{lmmse} channel estimation scheme and a novel uplink achievable \ac{se} considering the \ac{ici} caused by the \ac{pn}. Numerical results are used to evaluate the \ac{se} and to show that the proposed channel estimator provides higher \ac{se} compared to  both \acl{pna} and \acl{pnu} estimators stemmed from single-carrier systems.

\subsubsection*{{Notation}}
Lowercase and uppercase boldface letters, $\boldsymbol{x}$ and $\boldsymbol{X}$, denote column vectors and matrices respectively. The superscripts $^{\mathsf{T}}$, $^*$, $^{\mathsf{H}}$ and $^{\dagger}$ denote transpose, conjugate, conjugate transpose, and pseudo-inverse, respectively. Variables with the $\check{}$ mark, e.g., $\check{x}$, represents that they are in time-domain. The $n \times n$ identity matrix is $\mathbf{I}_n$. We use $\triangleq$ for definitions and $\text{diag}(\boldsymbol{x})$ for a diagonal matrix with $\boldsymbol{x}$ on the diagonal. The expected value of $\boldsymbol{x}$ is denoted by $\mathbb{E}\{\boldsymbol{x}\}$.

\section{System Model}
We consider a cell-free  \ac{ofdm} network comprising $L$ randomly distributed single-antenna \acp{ap}, which are connected to a \ac{cpu} via a fronthaul network and serve $K$ single-antenna \acp{ue}. Each \ac{ofdm} symbol consists of $N$ subcarriers with spacing $\Delta_f$. A \ac{cp} length of $N_{\text{CP}}$ is considered. The signal bandwidth is $W= N\Delta_f$ so that the sampling time is $T_s= 1/W$. The \ac{ofdm} symbol time is $T=NT_s=1/\Delta_f$.

\subsection{Block Fading Channel Model}
The time-domain channel between \ac{ue} $k$ and \ac{ap} $l$ is modeled as a $Q$-tap \acl{fir} filter $\check{\boldsymbol{h}}_{k,l} = [\check{h}_{k,l,0}, \cdots, \check{h}_{k,l,Q-1}]^{\mathsf{T}}$, where the filter length $Q$ is no longer than the multi-path delay spread $T_d$ normalized by the sampling time, i.e., $Q\leq \lceil {T_d}/{T_s} \rceil $. The corresponding channel in frequency-domain can be obtained by an $N$-point \ac{dft} on $\check{\boldsymbol{h}}_{k,l}$. This yields a correlated channel in the frequency domain. In this paper, we neglect this correlation but assume that the time-frequency resources are divided into coherence blocks where each channel is time-invariant and frequency-flat, considering the standard TDD multicarrier protocol of a canonical \acl{mmimo} network from~\cite{bjornson2017massive}. Each coherence block has a coherence time $T_c = \tau_c T$ and a coherence bandwidth $W_c=N_{c}\Delta f$, i.e., $\tau_c$ OFDM symbols and $N_c$ subcarriers. In total, the number of coherent channel uses is $W_cT_c=\tau_c N_c$, consisting of $\tau_p$ and $(\tau_c N_c- \tau_p) $ channel uses for pilot and data, respectively. Subcarriers in each \ac{ofdm} symbol are split by $R=\lceil N/N_c \rceil$ coherence blocks, where the subcarrier set in the coherence block $r \in \{1,\cdots,R\}$ is denoted by $\mathcal{R}_r=\{(r-1)N_{c},\cdots, rN_{c}-1 \}$. For an arbitrary coherence block $r$, the frequency-domain channel between UE $k$ and AP $l$ over subcarrier $n\in \mathcal{R}_r$ is denoted by $h_{k,l,n} \sim \mathcal{N}_{\mathbb{C}}(0,\beta_{k,l})$, where $\beta_{k,l}$ represents the large-scale fading coefficient. 
Notice that $h_{k,l,n_1} = h_{k,l,n_2}$, for ${n_1, n_2} \in \mathcal{R}_r$, while $\mathbb{E}\{h_{k,l,n_1} h_{k,l,n_2}\} = \mathbb{E}\{h_{k_1,l,n} h_{k_2,l,n}\}=0$, for $n_1, n_2 \notin \mathcal{R}_r$ and $k_1\neq k_2$.
\subsection{Phase Noise Model}
\vspace*{-0.cm}
Imperfect \acp{lo} at APs and UEs introduce \ac{pn}. The \ac{pn} $\check{\phi}_{l, m}^{(\tau)}$ and $\check{\varphi}_{k, m}^{(\tau)}$ from the \acp{lo} of AP $l$ and UE $k$, respectively, at time sample $m$ of OFDM symbol $\tau$, can be modeled as discrete-time Wiener processes~\cite{petrovic2007effects},
\begin{align}
\check{\phi}_{l, m}^{(\tau)}  = \check{\phi}_{l, m-1}^{(\tau)}+ \check{\delta}_m^{\phi},\,\, 
\check{\varphi}_{k, m}^{(\tau)}  =\check{\varphi}_{k, m-1}^{(\tau)}+\check{\delta}_m^{\varphi},
\end{align}
where $\check{\delta}_m^\phi\sim \mathcal{N}(0,\sigma_{\phi}^2)$ and $\check{\delta}_n^\varphi\sim \mathcal{N}(0,\sigma_{\varphi}^2)$. The increment variance of the \ac{pn} process is modeled as
    $\sigma_i^2=4 \pi^2 f_{\mathrm{c}}^2 \gamma_i T_{\mathrm{s}}$, for $i\in \{\phi, \varphi\}$, 
where $f_{\mathrm{c}}$ and $\gamma_i$ denote the carrier frequency and a constant describing the oscillator quality. Note that different APs and UEs may have different quality. The uplink received \ac{pn} from UE $k$ and AP $l$ at time-domain sample $m$ of OFDM symbol $\tau$ is combined as $\check{\theta}_{k,l,m}^{(\tau)} = \check{\phi}_{l, m}^{(\tau)} + \check{\varphi}_{k, m}^{(\tau)}$,
and its vector form for the whole OFDM symbol $\tau$ is denoted by $\check{\boldsymbol{\theta}}_{k,l}^{(\tau)}=[\check{\theta}_{k,l,0}^{(\tau)},\cdots,\check{\theta}_{k,l,N-1}^{(\tau)}]^{\mathsf{T}}$. Considering the \ac{cp} length $N_{\text{CP}}$, the phase noise at the first time-domain sample of $(\tau +1)$-th OFDM symbol can be modeled as
$\check{\theta}_{k,l,0}^{(\tau+1)} = \check{\theta}_{k,l,N-1}^{(\tau)} + \check{\delta}^{\text{CP}}$,
where $\check{\delta}^{\text{CP}}\sim \mathcal{N}(0,(N_{\text{CP}}+1)(\sigma_{\phi}^2+\sigma_{\varphi}^2))$.

For any OFDM symbol $\tau$, the multiplicative \ac{pn} exp(${j\check{\boldsymbol{\theta}}_{k,l}^{(\tau)}}) \in \mathbb{C}^N$ in time-domain is equivalent to the convolutional frequency-domain phase-drift vector
$\boldsymbol{J}_{k,l}^{(\tau)}\in \mathbb{C}^{N}$, whose $i$-th entry $J_{k,l,i}^{(\tau)}$, for $i=-N/2,\cdots,N/2-1$, is obtained by~\cite{petrovic2007effects}
\begin{align}
	J_{k,l,i}^{(\tau)} = \frac{1}{N}\sum\limits_{n=0}^{N-1}e^{j\theta^{(\tau)}_{k,l,n}} e^{-j2\pi ni/N}. \label{eq:J_k_l_i_definition}
\end{align} 
For $i=0$, the phase-drift $J_{k,l,0}^{(\tau)}={1}/{N}\sum_{n}e^{j\check{\theta}^{(\tau)}_{k,l,n}}$ is known as the \ac{cpe} ~\cite{petrovic2007effects} since it acts on all subcarriers. The other non-zero phase-drifts $J_{k,l,i}^{(\tau)}$ for $ i\neq 0$ lead to \ac{ici}. The correlation between $J_{k,l,i_1}^{(\tau_1)}$ and $J_{k,l,i_2}^{(\tau_2)}$ is calculated as~\cite{petrovic2007effects}
\begin{align}
&\mathbb{E}\{J_{k,l,i_1}^{(\tau_1)} J_{k,l,i_2}^{*(\tau_2)}\} \triangleq {B}_{i_1,i_2}^{(\tau_1 -\tau_2)}= \label{eq:B_matrix_compute}\\
&\frac{1}{N^2} \sum_{n_1=0}^{N-1}\sum_{n_2=0}^{N-1}  e^{ -\frac{\sigma_{\varphi}^2+\sigma_{\phi}^2}{2}\left(|(\tau_1-\tau_2)N + n_1- n_2 |\right)} e^{\frac{-j2\pi(n_1i_1- n_2i_2)}{N}}.\nonumber 
\end{align}


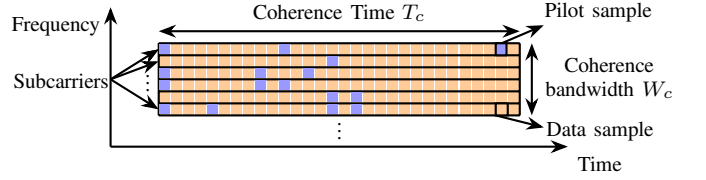
\begin{figure}[t]
    \centering
    \begin{tikzpicture}
[font=\footnotesize, draw=black!100, x=0.32cm,y=0.32cm
]
\usetikzlibrary {arrows.meta}

\draw[thick,-Stealth] (-2,-1.3) -- (17,-1.3) node[anchor=north west] {Time};
\draw[thick,-Stealth] (-2,-1.3) -- (-2,4.5) node[anchor=north east] {Frequency};

\foreach \x in {0,0.5,...,14.5}
    \foreach \y in {0,0.5,...,2.5}
        \fill[orange!40!white] (\x,\y) rectangle (\x+0.45,\y+0.45);

\draw[xstep=15,ystep=0.5,black, thick] (0,0) grid (15,3);

\foreach \xx in {1,1.5,...,10}
 \pgfmathrandominteger{\x}{0}{9}
  \pgfmathrandominteger{\y}{0}{5}
        \fill[blue!40!white] (\x,\y/2+0.05) rectangle (\x+0.452,\y/2+0.452);
        
\filldraw[blue!40!white,draw=black,thick] (14.0, 2.5) rectangle (14.45,3);
\filldraw[orange!40!white,draw=black,thick] (14., 0.0) rectangle (14.5,0.5);

  \draw[Stealth-Stealth,thick] (0,3.5)   -- (15,3.5) node[midway, above] {Coherence Time $T_c$};
    \draw[Stealth-Stealth,thick] (15.5,0)   -- (15.5,3.) node[midway,right=-0.5] {\begin{tabular}{c}
         Coherence\\
         bandwidth $W_c$ 
    \end{tabular}};

\foreach \ysub in{2.75,2.25,0.25}
    \draw[-Stealth,thick] (-2,1.5)   -- (0,\ysub) node[midway, left] {};
\draw[-Stealth] (-2,1.5)   -- (0,2.75) node[midway, below left=0. and 0.65] {Subcarriers};



\draw[-Stealth,thick] (14.25,3)   -- (16,3.75) node[midway, above right=0.1 and 0.5] {Pilot sample};
\draw[-Stealth,thick] (14,0.)   -- (16,-0.5) node[midway, below right=-0.4 and 0.7] {Data sample};

\node[]()at (-0.45,1.8) {\begin{tabular}{c} .\\ \end{tabular}};
\node[]()at (-0.45,1.5) {\begin{tabular}{c} .\\ \end{tabular}};
\node[]()at (-0.45,1.2) {\begin{tabular}{c} .\\ \end{tabular}};

\node[]()at (7.5,-0.2) {\begin{tabular}{c} .\\ \end{tabular}};
\node[]()at (7.5,-0.5) {\begin{tabular}{c} .\\ \end{tabular}};
\node[]()at (7.5,-0.8) {\begin{tabular}{c} .\\ \end{tabular}};



\end{tikzpicture}
     \vspace*{-1 \baselineskip}
    \caption{An example of a coherence block over the time-frequency plane with an arbitrary pilot distribution.}
    \label{fig:coherent_block}
\end{figure}
\subsection{Uplink Pilot Assignment} \label{sec:Uplink_pilot}

We assume the network utilizes a pilot book consisting of $\tau_p$ mutually orthogonal $\tau_p$-length frequency-domain pilot sequences, collected into a pilot set $\mathcal{S}_p = \{{\boldsymbol{s}}_{1}, {\boldsymbol{s}}_{2}, \ldots, {\boldsymbol{s}}_{\tau_p}\}$, with $\|{\boldsymbol{s}}_{t}\|^2=\tau_p$ and ${\boldsymbol{s}}_{t_1}^{\mathsf{H}} {\boldsymbol{s}}_{t_2}=0$ for $t_1 \neq t_2$. We assume each \ac{ue} uses the same pilot sequence for all coherence blocks. In each coherence block, the rest of the $(N_{c}\tau_c -\tau_p)$ channel uses are used for data transmission. When $K>\tau_p$, users have to share pilots, which causes pilot contamination~\cite{bjornson2017massive}.  
An arbitrary distribution of the $\tau_p$-length pilot in an coherence block can be implemented, as shown in Fig.~\ref{fig:coherent_block}. For an arbitrary coherence block $r$, the set of subcarrier and OFDM symbol indices used for pilot transmission is denoted by $\mathcal{N}_{p}  =\{n_1,\cdots,n_{\tau_p}\} \subseteq \mathcal{R}_r$, and $\mathcal{T}_{p}=\{\tau_1,\cdots,\tau_{\tau_p}\} \subseteq  \{1,\cdots,\tau_c\}$, respectively. The set of remaining subcarrier indices for data transmission is denoted by $\mathcal{N}_{d} = \mathcal{R}_r \backslash \mathcal{N}_{p}$. We assume all UEs use the same pilot distribution.

\ac{ue} $k$ is assigned with pilot ${\boldsymbol{s}}_{t_k} \in \mathbb{C}^{\tau_p}$ from $\mathcal{S}_p$, where we denote the index of the pilot  assigned to \ac{ue} k as $t_k \in\left\{1, \cdots, \tau_p\right\}$. \ac{ue} $k$ distributes its pilot ${\boldsymbol{s}}_{t_k}$ over pilot subcarriers $ n\in \mathcal{N}_{p}$ and OFDM symbols $\tau \in \mathcal{T}_p$, with each sample denoted by $s_{t_k,n}^{(\tau)}$. The $\tau$-th OFDM symbol of UE $k$ is ${\boldsymbol{s}}_{t_k}^{(\tau)} \in \mathbb{C}^{N}$, which consists of both pilot and data samples.  

\subsection{Signal Model}
\label{sec:signal_model}
For UE $k$, the $i$-th time-domain sample of $\tau$-th OFDM symbol is obtained by an IDFT on ${\boldsymbol{s}}_{t_k}^{(\tau)}$ as
		$\check{s}_{t_k,i}^{(\tau)} = \frac{1}{\sqrt{N}}\sum_{n=0}^{N-1} {s}_{t_k,n}^{(\tau)} e^{\frac{j2\pi ni}{N}}$,
and its vector form for the whole OFDM symbol $\tau$ is denoted by $\check{\boldsymbol{s}}^{(\tau)}_{t_k} \in \mathbb{C}^{N}$.  With the \ac{pn} $\text{diag}(e^{j\check{\boldsymbol{\theta}}_{k,l}^{(\tau)}})$, the time-domain received  signal $\check{\boldsymbol{y}}_{l}^{(\tau)} \in \mathbb{C}^{N}$ at AP $l$ for the $\tau$-th OFDM symbol is
\begin{align}
	\check{\boldsymbol{y}}_{l}^{(\tau)} = \sum\limits_{k=1}^{K} \sqrt{p_k} \text{diag}(e^{j\check{\boldsymbol{\theta}}_{k,l}^{(\tau)}})(\check{\boldsymbol{h}}_{k,l} \circledast \check{\boldsymbol{s}}^{(\tau)}_{t_k} ) + \check{\boldsymbol{\eta}}_{l}^{(\tau)},\label{eq:y_kl_TD}
\end{align}
where $p_k \geq 0$ is the transmit power of UE $k$, $\circledast$ denotes the circular convolution, $\check{\boldsymbol{\eta}}_{k,l}^{(\tau)} \sim \mathcal{N}_{\mathbb{C}}(\mathbf{0},\sigma^2\boldsymbol{I}_{N}) $ denotes the thermal noise. By applying a $N$-point DFT to both sides of \eqref{eq:y_kl_TD}, the frequency-domain received signal is
\begin{align}
		{\boldsymbol{y}}_{l}^{(\tau)} = \sum\limits_{k=1}^{K} \sqrt{p_k} \boldsymbol{J}_{k,l}^{(\tau)} \circledast ( {\boldsymbol{h}}_{k,l} \odot {\boldsymbol{s}}^{(\tau)}_{t_k} ) + {\boldsymbol{\eta}}_{l}^{(\tau)},\label{eq:y_kl_FD}
\end{align}
where $\odot$ denotes the Hadamard product, the thermal noise follows the same distribution ${\boldsymbol{\eta}}_{l}^{(\tau)} \sim \mathcal{N}_{\mathbb{C}}(\mathbf{0},\sigma^2\boldsymbol{I}_{N})$. The elements of the phase-drift  $\boldsymbol{J}_{k,l}^{(\tau)}$ are defined in~\eqref{eq:J_k_l_i_definition}.
The frequency-domain sample ${y}_{l,n}^{(\tau)}$ received over subcarrier $n$ can be decomposed as
\begin{align}
	{y}_{l,n}^{(\tau)} = \sum\limits_{k=1}^{K}&\Big(\sqrt{p_k}   {{s}}^{(\tau)}_{t_k,n}   {{h}}_{k,l,n}^{(\tau)}   +  \zeta_{k,l,n}^{(\tau)}\Big) +    {{\eta}}_{l,n}^{(\tau)}
 \label{eq:FD_recv_tau}
\end{align}
where
\begin{align} \label{eq:effective_channel}
    {{h}}_{k,l,n}^{(\tau)} \triangleq J_{k,l,0}^{(\tau)}{{h}}_{k,l,n}
\end{align}
is the effective channel (including the \ac{cpe}) and
\begin{align}
 \zeta_{k,l,n}^{(\tau)} = \sqrt{p_k} \sum\limits_{\substack{\jmath =0, \jmath \neq n}}^{N-1} {{s}}^{(\tau)}_{t_k,\jmath} J_{k,l,n-\jmath}^{(\tau)} {{h}}_{k,l,\jmath}
\end{align}
is the \ac{ici} component over subcarrier $n$ and \ac{ofdm} symbol $\tau$. Within an arbitrary coherence block $r$, the effective channels are the same only for subcarriers in the same \ac{ofdm} symbol,
\begin{align}
    {{h}}_{k,l,n_1}^{(\tau_1)} = {{h}}_{k,l,n_2}^{(\tau_1)} \neq {{h}}_{k,l,n_1}^{(\tau_2)}, \text{ for } \{n_1,n_2\} \in \mathcal{R}_{r}, \tau_1\neq \tau_2. \notag
\end{align}
Note that the \ac{cpe} coefficient $J_{k,l,0}^{(\tau)}$ is the same for all subcarriers of \ac{ofdm} symbol $\tau$. Other phase error coefficients $\{J_{k,l,i}^{(\tau)}, \, i\neq n\}$ introduce \acp{ici} on an arbitrary subcarrier $n$. 
Note also that the different \acp{cpe} break the pilot orthogonality. 
Finally, observe that utilizing the time-domain single-carrier \ac{pn} model to the frequency-domain OFDM system model~\eqref{eq:FD_recv_tau} leads to a mismatched system model as in~\cite{bjornson2015massive,papazafeiropoulos2021scalable,zheng2023asynchronous,jin2020spectral,ozdogan2019performance}.

\section{Channel Estimation and Uplink Data Transmission With Phase Noise}
We now derive a \acl{pna} channel estimator and an achievable uplink SE expression in the presence of \ac{pn} in cell-free massive MIMO networks.
\vspace*{-0.1cm}
\subsection{Channel Estimation With Phase Noise}
From~\eqref{eq:FD_recv_tau} and~\eqref{eq:effective_channel}, we see that, although the true channels follow a block fading model, the \ac{pn} make the effective channel $h_{k,l,n}^{(\tau)}$ change over each \ac{ofdm} symbol $\tau$. For each coherence block, we need to do the channel estimation for just one subcarrier of each OFDM symbol, instead of estimating the effective channel for every channel use as indicated by the single-carrier channel estimation methods in~\cite{bjornson2015massive,papazafeiropoulos2021scalable}. We now derive an estimator of the effective channel ${{h}}_{k,l,n}^{(\tau)}$ for any subcarrier $n \in \mathcal{R}_r$ and OFDM symbol $\tau \in \{1,\cdots,\tau_c\}$. The conventional \ac{mmse} estimation is hard to compute since the received signal and channel are not joint Gaussian distributed due to the \ac{pn}. Therefore, we derive a \ac{lmmse} estimator~\cite{kay1993fundamentals}. In the absence of \ac{pn}, the \ac{lmmse} estimator becomes the optimal \ac{mmse} estimator~\cite{bjornson2020scalable} since the received samples and channels are jointly Gaussian.
\begin{lemma}
The \ac{lmmse} estimate of ${h}_{k,l,n}^{(\tau)}$ based on ${\boldsymbol{y}}_{l} \triangleq [{y}_{l,n_1}^{(\tau_1)},\cdots,{y}_{l,n_{\tau_p}}^{(\tau_p)}]^{\mathsf{T}}$is
\begin{align}
    \hat{{h}}_{k,l,n}^{(\tau)} = \hat{{h}}_{k,l}^{(\tau)}=\sqrt{p_k}\beta_{k,l}{\boldsymbol{s}}_{t_k}^{\mathsf{H}} \boldsymbol{B}_{0,0}^{(\tau)} \boldsymbol{\Psi}_{l}^{-1} {\boldsymbol{y}}_{l}.
\end{align}
where
\begin{align} 
\boldsymbol{B}_{0,0}^{(\tau)}&=\mathrm{diag}\left(\Big[{B}_{0,0}^{(\tau-1)} ,\cdots, {B}_{0,0}^{(\tau-\tau_p)} \Big]^{\mathsf{T}}\right) \label{eq:LMMSE_weights_matrix}
\end{align}
\begin{align}
\boldsymbol{\Psi}_{l} &= \sum\nolimits_{k=1}^{K}p_{k}\beta_{k,l} \boldsymbol{\Phi}_{k} + \boldsymbol{Z}^{\text{ICI}}_{l} + \sigma^2 \mathbf{I}_{\tau_p} 
\\
[\boldsymbol{\Phi}_{k}]_{\tau_1, \tau_2} &= {s}^{(\tau_1)}_{t_k,n_1} {s}^{*(\tau_2)}_{t_k,n_2} {B}_{0,0}^{(\tau_1-\tau_2)}.
\end{align}
Here   ${B}_{i_1,i_2}^{(\tau_1-\tau_2)}$ and $\boldsymbol{Z}^{\text{ICI}}_{l}$  are defined in~\eqref{eq:B_matrix_compute} and ~\eqref{eq:Zeta_ICI}.
\end{lemma}

\begin{proof}
See Appendix \ref{sec:Appendix_A}.
\end{proof}
With the \ac{lmmse} estimation, 
the channel estimation $\hat{{h}}_{k,l,m}^{(\tau)}$ is with zero mean and variance $ \epsilon_{k,l,m}^{(\tau)} \triangleq p_k \beta_{k,l}^2 {\boldsymbol{s}}_{t_k,n}^{\mathsf{H}} \boldsymbol{B}_{0,0}^{(\tau)} \boldsymbol{\Psi}_{l,n}^{-1} \boldsymbol{B}_{0,0}^{\mathsf{H},(\tau)} {\boldsymbol{s}}_{t_k,n}^{\mathsf{}} $, and the channel estimation error  $\tilde{{h}}_{k,l,n}^{(\tau)}\triangleq {{h}}_{k,l,n}^{(\tau)} - \hat{{h}}_{k,l,n}^{(\tau)}$ is with zero mean and variance $ c_{k,l,m}^{(\tau)} \triangleq \beta_{k,l} - \epsilon_{k,l,m}^{(\tau)} $.

\vspace{-0.2cm}
\subsection{Uplink Data Transmission} 
%
%
%
To meet the scalability requirement of CF \acl{mmimo}, we assume that an arbitrary AP $l$ only serves a subset of UEs~\cite{bjornson2020scalable}. 
We denote the set of UEs served by AP $l$ by~\cite{bjornson2020scalable}
\begin{align}
    \mathcal{D}_{l} = \{k : d_{k,l}=1, k \in \{1,\cdots, K\} \},
\end{align}
where $d_{k,l} \in \{0,1\}$ defines  whether UE $k$ and AP $l$ communicate to each other according to the \ac{dcc} framework~\cite{bjornson2011optimality}. For an arbitrary AP $l$, the cardinality $|\mathcal{D}_l|$ is constant as $K \rightarrow \infty$ to satisfy the scalability of a CF \acl{mmimo} network.
\begin{figure*}[!t]
\begin{equation}
\text{SINR}_{k,n}^{(\tau)} = \frac{p_k \left|\mathbb{E}\left\{\boldsymbol{v}_{k,n}^{\mathsf{H},(\tau)} \boldsymbol{D}_k \boldsymbol{h}_{k,n}^{(\tau)}\right\}\right|^2}
{\sum\nolimits_{i=1 }^{K} p_i  \mathbb{E}\Big\{\left|\boldsymbol{v}_{k,n}^{\mathsf{H},(\tau)} \boldsymbol{D}_k \boldsymbol{h}_{i,n}^{(\tau)}\right|^2\Big\} - p_k \left|\mathbb{E}\Big\{\boldsymbol{v}_{k,n}^{\mathsf{H}, (\tau)} \boldsymbol{D}_k \boldsymbol{h}_{k,n}^{(\tau)}\right\}\Big|^2 + \rho^{\text{ICI},(\tau)}_{k,n} + \sigma^2\mathbb{E}\Big\{\left| \boldsymbol{D}_k \boldsymbol{v}_{k,n}^{\mathsf{H},(\tau)} \right|^2\Big\}
}\label{eq:SINR_lo_UatF}
\end{equation}
\hrule
\end{figure*}

For \ac{ue} $k$, let $s_{k,n}^{(\tau)} \in \mathbb{C}$ denote the uplink data sample over subcarrier $n \in \mathcal{N}_{\text{d}}$ and OFDM symbol $\tau \in\{1,\cdots,\tau_c\}$ with zero mean and power $p_k$. The received signals from all \acp{ap} are collected at the CPU as
\begin{align}
    \boldsymbol{y}_{n}^{(\tau)} = \sum_{k=1}^{K} \boldsymbol{h}_{k,n}^{(\tau)} s_{k,n}^{(\tau)}  + \sum_{k=1}^{K} \boldsymbol{\zeta}_{k,n}^{(\tau)} + \boldsymbol{\eta}_{n}^{(\tau)},
\end{align} 
where the concatenate channel, ICI, and thermal noise between UE $k$ and all $L$ APs over subcarrier $n$ of the $\tau$-th OFDM symbol are denoted as $\boldsymbol{h}_{k,n}^{(\tau)}=[h_{k,1,n}^{(\tau)},\cdots, h_{k,L,n}^{(\tau)}]^{\mathsf{T}}$, $\boldsymbol{\zeta}_{k,n}^{(\tau)}=[\zeta_{k,1,n}^{(\tau)},\cdots, \zeta_{k,L,n}^{(\tau)}]^{\mathsf{T}}$, and $\boldsymbol{\eta}_{n}^{(\tau)}  \sim \mathcal{N}_{\mathbb{C}}(\mathbf{0},\sigma^2\boldsymbol{I}_{L})$. The ICI $\zeta_{k,l,n}^{(\tau)}$ is defined in in~\eqref{eq:FD_recv_tau}.

The collective channel estimates and channel estimation error are defined by $\hat{\boldsymbol{h}}_{k,n}^{(\tau)}=[\hat{h}_{k,1,n}^{(\tau)},\cdots, \hat{h}_{k,L,n}^{(\tau)}]^{\mathsf{T}}$ and $\tilde{\boldsymbol{h}}_{k,n}^{(\tau)}=[\tilde{h}_{k,1,n}^{(\tau)},\cdots, \tilde{h}_{k,L,n}^{(\tau)}]^{\mathsf{T}}$ with zero means and variances $\text{diag}([\epsilon_{k,1,n}^{(\tau)} ,\cdots, \epsilon_{k,L,n}^{(\tau)}]^{\mathsf{T}})$ and  $\boldsymbol{C}_{k,n}^{(\tau)}=\text{diag}([c_{k,1,n}^{(\tau)} ,\cdots, c_{k,L,n}^{(\tau)}]^{\mathsf{T}})$, respectively. 

According to the \ac{dcc} network, only a subset of the APs participant in the signal detection.  The CPU selects a receive combining scalar $v_{k,l,n}^{(\tau)}$ for an arbitrary UE $k$ and AP $l$, and completes the estimate of $s_{k,n}^{(\tau)}$  by computing the summation
\begin{align}
    \hat{s}_{k,n}^{(\tau)}
    &= \underbrace{\boldsymbol{v}_{k,n}^{\mathsf{H},(\tau)} \boldsymbol{D}_k \boldsymbol{h}_{k,n}^{(\tau)} s_{k,n}^{(\tau)}}_{\text{Desired signal}} + \underbrace{\sum\nolimits_{\substack{i\neq k}}^{K}\boldsymbol{v}_{k,n}^{\mathsf{H},(\tau)} \boldsymbol{D}_k \boldsymbol{h}_{i,n}^{(\tau)} s_{i,n}^{(\tau)}}_{\text{Inter-user interference}}  \nonumber\\
    &+ \underbrace{\sum\nolimits_{i=1}^{K} \boldsymbol{v}_{k,n}^{\mathsf{H},(\tau)} \boldsymbol{D}_k \boldsymbol{\zeta}_{i,n}^{(\tau)}}_{\text{ICI}} 
    + \boldsymbol{v}_{k,n}^{\mathsf{H},(\tau)} \boldsymbol{D}_k \boldsymbol{\eta}_{n}^{(\tau)},
    \label{eq:s_k_n_level4}
\end{align}
where $\boldsymbol{v}_{k,n}^{(\tau)}=[v_{k,1,n}^{(\tau)} ,\cdots, v_{k,L,n}^{(\tau)}]^{\mathsf{T}}$ denotes the collective combining vector and $\boldsymbol{D}_k = \text{diag}([d_{k,1},\cdots d_{k,L}]^{\mathsf{T}})$ denotes a diagonal matrix.
\subsection{Uplink Spectral Efficiency}
The ergodic capacity is unknown for this setup with the \ac{pn}. We follow the \ac{uatf} bound in \acl{mmimo}~\cite[Th. 4.4]{bjornson2017massive} and also in~\cite{demir2021foundations,nayebi2016performance,bashar2019uplink} for CF \acl{mmimo} to give an achievable \ac{se} expression.
\begin{proposition}
An achievable SE of UE $k$ at data subcarrier $n \in \mathcal{N}_d$ is given by
\begin{align}
    \text{SE}_{k,n}^{\text{}} =\frac{1}{\tau_c} \sum_{\tau=1}^{\tau_c} \log_2(1+\text{SINR}_{k,n}^{ (\tau)}),\label{eq:SE_UatF_PN_aware}
\end{align}
where $\text{SINR}_{k,n}^{ (\tau)}$ is the eﬀective \ac{sinr} of UE $k$ over subcarrier $n$, given in~\eqref{eq:SINR_lo_UatF} with the ICI term $\rho_{k,n}^{\text{ICI},(\tau)}$ being defined in \eqref{eq:rho_ICI}.
\end{proposition}

\begin{proof}
    See Appendix~\ref{sec:Appendix_SINR_lo_UatF}.
\end{proof}

Note that the  \ac{sinr} in~\eqref{eq:SE_UatF_PN_aware} is different for each OFDM symbol $\tau$ because of \ac{cpe} introduced by the time-domain \ac{pn}. 

The SE expression in~\eqref{eq:SE_UatF_PN_aware} can be computed numerically for any combiner  $\boldsymbol{v}_{k,n}^{(\tau)}$ using Monte Carlo methods. 
In the context of cell-free massive MIMO, common combiners are represented by \ac{mr}, \ac{lpmmse}, \ac{mmse}, and \ac{pmmse} combinings, given as~\cite[Eqs. (19),(29), (23), (20)]{bjornson2020scalable}
\begin{align}
    \boldsymbol{v}_{k,n}^{\text{MR},(\tau)} &= \boldsymbol{D}_k \hat{\boldsymbol{h}}_{k,n}^{(\tau)} 
    \\
        \boldsymbol{v}_{k,l,n}^{\text{LP-MMSE},(\tau)} &=  
     p_k \Big(\sum\limits_{i\in \mathcal{D}_l}^{} p_i  |\hat{{h}}_{i,l,n}^{(\tau)}|^2 + c_{i,l,n}^{(\tau)} + \sigma^2 \Big)^{\dagger}  \hat{{h}}_{i,l,n}^{(\tau)}\label{eq:v_k_LPMMSE}\\
    \boldsymbol{v}_{k,n}^{\text{PMMSE},(\tau)} &=  
     p_k \Big(\sum\limits_{i\in \mathcal{P}_k}^{} p_i\hat{\boldsymbol{H}}_{i,n}^{D,(\tau)}    + \boldsymbol{Z}_{i,n}^{\prime (\tau)} \Big)^{\dagger} \boldsymbol{D}_k \hat{\boldsymbol{h}}_{i,n}^{(\tau)}\label{eq:v_k_PMMSE}\\
    \boldsymbol{v}_{k,n}^{\text{MMSE},(\tau)} 
    &= p_k \Big(\sum\limits_{i=1}^{K} p_i \hat{\boldsymbol{H}}_{i,n}^{D,(\tau)}  + \boldsymbol{Z}_{i,n}^{(\tau)} \Big)^{\dagger} \boldsymbol{D}_k \hat{\boldsymbol{h}}_{i,n}^{(\tau)},
    \label{eq:v_k_MMSE}
\end{align}
where $\hat{\boldsymbol{H}}_{i,n}^{D,(\tau)} =\boldsymbol{D}_k \hat{\boldsymbol{h}}_{i,n}^{(\tau)} \hat{\boldsymbol{h}}_{i,n}^{\mathsf{H},(\tau)} \boldsymbol{D}_k$, $\mathcal{P}_k = \{i: \boldsymbol{D}_k \boldsymbol{D}_i \neq \mathbf{0}_L  \}$, $\boldsymbol{Z}_{i,n}^{\prime (\tau)}=\boldsymbol{D}_k\big(\sum_{i\in \mathcal{P}_k}^K p_i \boldsymbol{C}_{i,n}^{(\tau)}+\sigma_{\text{}}^2 \boldsymbol{I}_{L}\big) \boldsymbol{D}_k$, and $\boldsymbol{Z}_{i,n}^{(\tau)}=\boldsymbol{D}_k\big(\sum_{i=1}^K p_i \boldsymbol{C}_{i,n}^{(\tau)}+\sigma_{\text{}}^2 \boldsymbol{I}_{L}\big) \boldsymbol{D}_k$. We refer to a combining scheme \acl{pna} or \acl{pnu}, depending on the usage of \acl{pna} or \acl{pnu} channel estimators, respectively. 


\section{Numerical Results} 
Numerical results are now given to show the advantage of the proposed \acl{pna} \ac{lmmse} channel estimator over other estimators in a cell-free  \ac{ofdm} network. 

\subsection{Scenario}
We consider a simulation scenario where $L=200$ \ac{ap} and $K=10$ \acp{ue} are independently and uniformly distributed in a $1\times 1$\,km square, all equipped with single-antenna. This approximates an infinitely large network with $200$ antennas/km$^2$ and $10$ \acp{ue}/km$^2$. We use the same propagation model as in~\cite{bjornson2020scalable} with spatially correlated fading. We assume that the coherence time and bandwidth are $T_c=1$ ms and $W_c=180$ kHz, respectively, which fits an coherence block setup of $N_{c}=12$ subcarriers with subcarrier spacing $\Delta f=15$ kHz and $\tau_c=15$ OFDM symbols. In total, each coherence block contains $180$ samples with $\tau_p=12$ pilot samples and $(N_{c}\tau_c-\tau_p)=168$ data samples. We distribute each pilot sequence to the first subcarrier of $\tau_p=12$ OFDM symbols, i.e., $\mathcal{N}_p=\{0\}$ and $\mathcal{T}_p=\{1,\cdots,12\}$. Each OFDM symbol contains $N=1200$ subcarriers which leads to a signal bandwidth $W=18$ MHz and symbol time $T_s\approx 5.6\times 10^{-8}$s. We assume all \acp{ap} and \acp{ue} have the same quality \acp{lo}, with the same level but worse quality than that in~\cite{bjornson2015massive,papazafeiropoulos2021scalable}, i.e., $\gamma_{\phi}=\gamma_{\varphi}=4\times 10^{-17}$, which leads to a \ac{pn} variance $\sigma^2_{\phi}=\sigma^2_{\varphi}\approx 3.5\times 10^{-4}$ by setting $f_c=2$ GHz.

\subsection{Results and Discussion}
Fig.~\ref{fig:SE_vs_OFDM_Symbol} illustrates the relation of the uplink \acp{se} per \ac{ue} to
channel uses in the first coherence block for two combining schemes (MR and MMSE) with the same generative model defined in this paper but three different channel estimators: 
\acl{pnu} \ac{mmse}~\cite{bjornson2020scalable} (marked by circles), single-carrier \acl{pna} \ac{lmmse}~\cite{papazafeiropoulos2021scalable} (named by PNA-SC and marked by triangles), and the proposed OFDM \acl{pna} \ac{lmmse} (named by PNA-OFDM and marked by squares).  We save  space to not show the results of LP-MMSE and PMMSE combiners as they imply similar messages. The ideal case of MR and MMSE combinings with no \ac{pn}
 are also shown. We notice that both combining schemes with both \acl{pna} estimators have substantial \ac{se} gains over the same combining schemes with conventional \acl{pnu} MMSE estimator for channel uses within the $\tau_p=12$ pilot \ac{ofdm} symbols, i.e., channel use $\leq 144$, while the SEs drop quickly as the channel use $> 144$, where the channel aging effects caused by the \ac{pn} become stronger because the chosen pilot distribution has no pilot samples for \ac{ofdm} symbol $\tau>\tau_p$. The proposed \acl{pna} \ac{lmmse} estimator perform better than the single-carrier \acl{pna} \ac{lmmse}, especially for the \ac{mmse} combiner. It is interesting to see only the combining schemes with the \acl{pnu} estimator have a convex shape, i.e., SEs are better in the middle than the beginning and the end. This can be explained that the \acl{pna} estimator gives different weights for pilot samples in different OFDM symbols, i.e., $\boldsymbol{B}_{0,0}^{(\tau)}$ in~\eqref{eq:LMMSE_weights_matrix}, while the \acl{pnu} MMSE estimator gives the same weights, which happen to fit better for channel uses around $(\tau_pN_c)/2=72$.

Fixing the same setup as in Fig.~\ref{fig:SE_vs_OFDM_Symbol} but varying the number of \acp{ue} $K$, we evaluate the corresponding \acp{se} at channel use $60$ in Fig.~\ref{fig:SE_vs_K}. We notice that all \acp{se} decrease as $K$ grows because both the \ac{ici} and inter-user interference caused by the \ac{pn} increase with $K$ as we indicate in Section~\ref{sec:Uplink_pilot}. The \ac{mmse} combining with the proposed \acl{pna} \ac{lmmse} estimator performs the best but eventually degrades to the same as the \ac{mmse} (Aware-SC) combining due to strong pilot contamination.

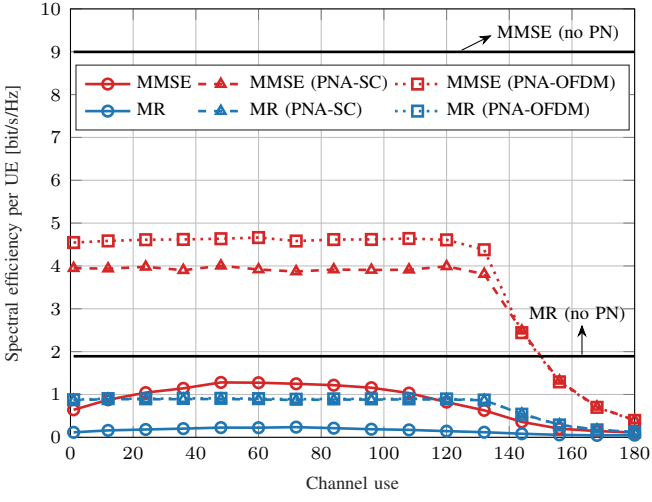
\begin{figure}[t]
    \centering
    \vspace*{-0 \baselineskip}
	\begin{tikzpicture}[font=\scriptsize]
\definecolor{color0}{rgb}{0.12156862745098,0.466666666666667,0.705882352941177}
\definecolor{color1}{rgb}{1,0.498039215686275,0.0549019607843137}
\definecolor{color2}{rgb}{0.172549019607843,0.627450980392157,0.172549019607843}
\definecolor{color3}{rgb}{0.83921568627451,0.152941176470588,0.156862745098039}
\definecolor{color4}{rgb}{0.580392156862745,0.403921568627451,0.741176470588235}
\definecolor{color5}{rgb}{0.549019607843137,0.337254901960784,0.294117647058824}
\definecolor{color6}{rgb}{0.890196078431372,0.466666666666667,0.76078431372549}
\definecolor{color7}{rgb}{0.737254901960784,0.741176470588235,0.133333333333333}

\begin{axis}[%
width=7.5cm,
height=5.7cm,
at={(1.011in,0.642in)},
scale only axis,
xmin=0,
xmax=180,
xlabel style={font=\color{white!15!black}},
xlabel={\scriptsize Channel use},
ymin=0,
ymax=10,
ylabel style={font=\color{white!15!black}},
ylabel={\scriptsize Spectral efficiency per UE [bit/s/Hz]},
axis background/.style={fill=white},
title style={font=\bfseries},
xmajorgrids,
ymajorgrids,
legend style={at={(0.5,0.87)}, anchor=north, legend cell align=left, draw=white!15!black},
legend columns=3,
ytick distance={1}
]

\addplot [color=color3, line width=1pt, mark=o, mark options={solid, fill opacity=0}]
  table[row sep=crcr]{%
1	0.642320411951921\\
12	0.876218363852657\\
24	1.04312233714234\\
36	1.14293864300465\\
48	1.2809535350124\\
60	1.27470684722615\\
72	1.24932393991151\\
84	1.21674621767509\\
96	1.15813308199387\\
108	1.03274387306323\\
120	0.82145084080571\\
132	0.629731204380451\\
144	0.359849705285755\\
156	0.202185551640939\\
168	0.145013434904593\\
180	0.111699012797653\\
};
\addlegendentry{\scriptsize MMSE}

\addplot [color=color3, dashed, line width=1pt, mark=triangle, mark options={solid, color3}]
  table[row sep=crcr]{%
1	3.94778148231953\\
12	3.9375759491392\\
24	3.97864576937671\\
36	3.9009428463476\\
48	4.00441471261337\\
60	3.9173031533269\\
72	3.87021471574344\\
84	3.9160328387626\\
96	3.90627845886249\\
108	3.91184741204198\\
120	3.99024260849791\\
132	3.81097237982507\\
144	2.48743283518225\\
156	1.31774281701019\\
168	0.710668465099348\\
180	0.400535962778692\\
};
\addlegendentry{\scriptsize MMSE (PNA-SC)}

\addplot [color=color3, dotted, line width=1pt, mark=square, mark options={solid, fill opacity=0}]
  table[row sep=crcr]{%
1	4.54651301434491\\
12	4.58454012369535\\
24	4.61232141473384\\
36	4.61914462301599\\
48	4.63526241500425\\
60	4.66263550116481\\
72	4.58170630796274\\
84	4.61599307716861\\
96	4.61834351843125\\
108	4.64046437607206\\
120	4.60907092632158\\
132	4.37805883967515\\
144	2.44299289713311\\
156	1.29424440554473\\
168	0.700088556794035\\
180	0.393922808147848\\
};
\addlegendentry{\scriptsize MMSE (PNA-OFDM)}

\addplot [color=color0, line width=1pt, mark=o, mark options={solid, fill opacity=0}]
  table[row sep=crcr]{%
1	0.116246001401522\\
12	0.161507520287562\\
24	0.182401648387922\\
36	0.202893235915402\\
48	0.226523852837164\\
60	0.224577463849236\\
72	0.237626982023782\\
84	0.213647144422007\\
96	0.189131709401636\\
108	0.173228255550383\\
120	0.14379453813342\\
132	0.118827765216967\\
144	0.0833315555995425\\
156	0.0565785180110599\\
168	0.0449696507105068\\
180	0.0501702529486685\\
};
\addlegendentry{\scriptsize MR}

\addplot [color=color0, dashed, line width=1pt, mark=triangle, mark options={solid, fill opacity=0}]
  table[row sep=crcr]{%
1	0.865737913154621\\
12	0.891014031164551\\
24	0.888031976852213\\
36	0.891350563397722\\
48	0.89454881748121\\
60	0.883410945865308\\
72	0.871518456145022\\
84	0.881153924223839\\
96	0.879583401070633\\
108	0.885397686672748\\
120	0.881956296335488\\
132	0.865254176358221\\
144	0.547661093453543\\
156	0.293867575802791\\
168	0.175622214978817\\
180	0.124245263525132\\
};
\addlegendentry{\scriptsize MR (PNA-SC)}

\addplot [color=color0, dotted, line width=1pt, mark=square, mark options={solid, fill opacity=0}]
  table[row sep=crcr]{%
1	0.877416356022781\\
12	0.90290687149126\\
24	0.899060500468328\\
36	0.907668553886638\\
48	0.909184900882173\\
60	0.902315184403527\\
72	0.888527961832466\\
84	0.896785747168448\\
96	0.896861215590468\\
108	0.897533966987854\\
120	0.895184422472021\\
132	0.870769589034787\\
144	0.539055326063333\\
156	0.288545398388139\\
168	0.173996202087842\\
180	0.123625444506468\\
};
\addlegendentry{\scriptsize MR (PNA-OFDM)}

\addplot [color=black, line width=1pt]
  table[row sep=crcr]{%
1	8.99639822357104\\
12	8.99639822357104\\
24	8.99639822357104\\
36	8.99639822357104\\
48	8.99639822357104\\
60	8.99639822357104\\
72	8.99639822357104\\
84	8.99639822357104\\
96	8.99639822357104\\
108	8.99639822357104\\
120	8.99639822357104\\
132	8.99639822357104\\
144	8.99639822357104\\
156	8.99639822357104\\
168	8.99639822357104\\
180	8.99639822357104\\
};

\addplot [color=black, solid, line width=1.0pt]
  table[row sep=crcr]{%
1	1.89217045688083\\
13	1.89217045688083\\
25	1.89217045688083\\
37	1.89217045688083\\
49	1.89217045688083\\
61	1.89217045688083\\
73	1.89217045688083\\
85	1.89217045688083\\
97	1.89217045688083\\
109	1.89217045688083\\
121	1.89217045688083\\
133	1.89217045688083\\
145	1.89217045688083\\
157	1.89217045688083\\
169	1.89217045688083\\
181	1.89217045688083\\
};

\draw[-Stealth] (5.2cm, 5.15cm)   -- (5.6cm, 5.35cm) node[midway, above right=-0.1cm and 0.15cm] {\scriptsize MMSE (no PN)};

\draw[-Stealth] (6.8cm, 1.1cm)   -- (6.8cm, 1.5cm) node[midway, above left=0.1cm and -0.7cm] {\scriptsize MR (no PN)};

\end{axis}

\end{tikzpicture}%
    \vspace*{0.25 \baselineskip}
    \caption{Uplink SE per UE versus the OFDM symbol $\tau$ for two combining schemes with the proposed \acl{pna} \ac{lmmse}, single-carrier \acl{pna} \ac{lmmse}, and \acl{pnu} \ac{mmse} estimators for $\sigma^2_{\phi}=\sigma^2_{\varphi}= 3.5\times 10^{-4}$.}
    \label{fig:SE_vs_OFDM_Symbol}
     \vspace*{-0.5 \baselineskip}
\end{figure}
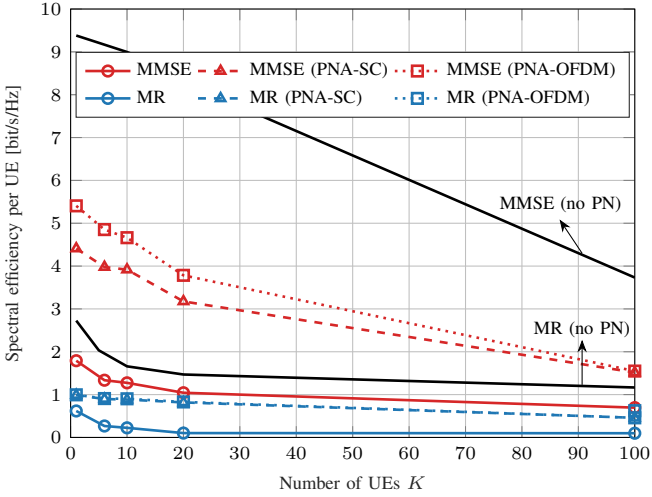
\begin{figure}[t]
    \centering
        \vspace*{-0 \baselineskip}
	\begin{tikzpicture}[font=\scriptsize]
\definecolor{color0}{rgb}{0.12156862745098,0.466666666666667,0.705882352941177}
\definecolor{color1}{rgb}{1,0.498039215686275,0.0549019607843137}
\definecolor{color2}{rgb}{0.172549019607843,0.627450980392157,0.172549019607843}
\definecolor{color3}{rgb}{0.83921568627451,0.152941176470588,0.156862745098039}
\definecolor{color4}{rgb}{0.580392156862745,0.403921568627451,0.741176470588235}
\definecolor{color5}{rgb}{0.549019607843137,0.337254901960784,0.294117647058824}
\definecolor{color6}{rgb}{0.890196078431372,0.466666666666667,0.76078431372549}
\definecolor{color7}{rgb}{0.737254901960784,0.741176470588235,0.133333333333333}

\begin{axis}[%
width=7.5cm,
height=5.7 cm,
at={(0,0)},
scale only axis,
xmin=0,
xmax=100,
xlabel style={font=\color{white!15!black}},
xlabel={\scriptsize Number of UEs $K$},
ymin=0,
ymax=10,
ylabel style={font=\color{white!15!black}},
ylabel={\scriptsize Spectral efficiency per UE [bit/s/Hz]},
axis background/.style={fill=white},
title style={font=\bfseries},
xmajorgrids,
ymajorgrids,
legend style={at={(0.5,0.9)}, anchor=north, legend cell align=left, draw=white!15!black},
legend columns=3,
ytick distance={1},
xtick distance={10}
]

\addplot [color=color3, line width=1pt, mark=o, mark options={solid, fill opacity=0}]
  table[row sep=crcr]{%
1	1.79127578437014\\
6	1.33576770086422\\
10	1.27470684722615\\
20	1.04324985602345\\
100	0.696061468465811\\
};
\addlegendentry{\scriptsize MMSE}

\addplot [color=color3, dashed, line width=1pt, mark=triangle, mark options={solid, color3}]
  table[row sep=crcr]{%
1	4.40921119293327\\
6	3.97809968280554\\
10	3.9173031533269\\
20	3.17493757372631\\
100	1.51626745904381\\
};
\addlegendentry{\scriptsize MMSE (PNA-SC)}

\addplot [color=color3, dotted, line width=1pt, mark=square, mark options={solid, fill opacity=0}]
  table[row sep=crcr]{%
1	5.40601976170866\\
6	4.8520003448414\\
10	4.66263550116481\\
20	3.78411779928758\\
100	1.54910851289621\\
};
\addlegendentry{\scriptsize MMSE (PNA-OFDM)}

\addplot [color=color0, line width=1pt, mark=o, mark options={solid, fill opacity=0}]
  table[row sep=crcr]{%
1	0.617962179152191\\
6	0.26725735078873\\
10	0.224577463849236\\
20	0.101767631551631\\
100	0.09840643004294\\
};
\addlegendentry{\scriptsize MR}

\addplot [color=color0, dashed, line width=1pt, mark=triangle, mark options={solid, fill opacity=0}]
  table[row sep=crcr]{%
1	0.989763147219384\\
6	0.897217329098751\\
10	0.883410945865308\\
20	0.817602082463886\\
100	0.458310490585065\\
};
\addlegendentry{\scriptsize MR (PNA-SC)}

\addplot [color=color0, dotted, line width=1pt, mark=square, mark options={solid, fill opacity=0}]
  table[row sep=crcr]{%
1	0.995127852044789\\
6	0.907781769896994\\
10	0.902315184403527\\
20	0.828984771658365\\
100	0.458310490585065\\
100	0.624574368338142\\
};
\addlegendentry{\scriptsize MR (PNA-OFDM)}

\addplot [color=black, line width=1pt]
  table[row sep=crcr]{
1	9.38132826884431\\
5	9.21091874271483\\
10	8.99639822357104\\
12	8.8946400377732\\
15	8.7621054932628\\
20	8.29296500233821\\
100	3.73193245641315\\
};

\addplot [color=black, line width=1.0pt, solid]
  table[row sep=crcr]{%
1	2.72688806319375\\
5	2.03342274912325\\
10	1.65967449943746\\
20	1.46924816657056\\
100 1.16607505195431\\
};

\draw[-Stealth] (6.8cm, 2.45cm)   -- (6.5cm, 2.9cm) node[midway, above left=0.2cm and -0.8cm] {\scriptsize MMSE (no PN)};

\draw[-Stealth] (6.8cm, 0.7cm)   -- (6.8cm, 1.3cm) node[midway, above = 0.21cm] {\scriptsize MR (no PN)};

\end{axis}

\end{tikzpicture}%
     \vspace*{0.25 \baselineskip}
    \caption{Uplink SE per UE versus the number of \acp{ue} $K$ for two combining schemes with the proposed \acl{pna} \ac{lmmse}, single-carrier \acl{pna} \ac{lmmse}, and \acl{pnu} \ac{mmse} estimators for $\sigma^2_{\phi}=\sigma^2_{\varphi}= 3.5\times 10^{-4}$.}
    \label{fig:SE_vs_K}
         \vspace*{0.25 \baselineskip}
\end{figure}

\section{Conclusion}
In this paper, we derived a signal model of \ac{pn}-impaired cell-free massive MIMO \ac{ofdm} networks and proposed a novel \ac{pn}-aware \ac{lmmse} channel estimator, which estimates any aging channel in the coherence block caused by the \ac{pn}. We derived a new uplink achievable \ac{se} expression considering the \ac{ici} from all \acp{ue}. Numerical results  demonstrate the advantage of the proposed \acl{pna} \ac{lmmse} estimator over both \acl{pnu} and a single-carrier \acl{pna} estimators for different receiving schemes.
\appendices{}
\section{Proof of Lemma 1}
\label{sec:Appendix_A}
The general expression for the \ac{lmmse}  estimator is~\cite{kay1993fundamentals}
$\hat{{h}}_{k,l,n}^{(\tau)}=\mathbb{E}  \{{{h}}_{k,l,n}^{(\tau)} {\boldsymbol{y}}_{l}^{\mathsf{H}} \} \big(\mathbb{E}\{ {\boldsymbol{y}}_{l} {\boldsymbol{y}}_{l}^{\mathsf{H}}\}\big)^{-1} {\boldsymbol{y}}_{l}$,
where we have 
\begin{align}
    \mathbb{E}  \{{{h}}_{k,l,n}^{(\tau)} & {\boldsymbol{y}}_{l}^{\mathsf{H}} \}= \sqrt{p_k}\beta_{k,l} \times \nonumber\\ 
    &\big[{{s}}^{*(\tau_1)}_{t_k,n_1}\mathbb{E}\{J_{k,l,0}^{(\tau)}J_{k,l,0}^{*(\tau_1)}\},\cdots,{{s}}^{*(\tau_p)}_{t_k,n_{\tau_p}}\mathbb{E}\{J_{k,l,0}^{(\tau)}J_{k,l,0}^{*(\tau_p)}\}   \big]\nonumber\\
    &\quad\;\;\, = \sqrt{p_k}\beta_{k,l}{\boldsymbol{s}}_{t_k}^{\mathsf{H}} \boldsymbol{B}_{0,0}^{(\tau)},
\end{align}
where $\boldsymbol{B}_{0,0}^{(\tau)}=\text{diag}[\mathbb{E}\{J_{k,l,0}^{(\tau)}J_{k,l,0}^{*(\tau_1)}\},\cdots, \mathbb{E}\{J_{k,l,0}^{(\tau)}J_{k,l,0}^{*(\tau_{p})}\}]$ and its $\tau^{\prime}$ element ${B}_{0,0}^{(\tau-\tau^{\prime})} \triangleq \mathbb{E}\{J_{k,l,0}^{(\tau)}J_{k,l,0}^{*(\tau^{\prime})}\}$ is calculated following~\eqref{eq:B_matrix_compute}. 

Furthermore, we have
\begin{align}
    &\mathbb{E}  \left\{{\boldsymbol{y}}_{l} {\boldsymbol{y}}_{l}^{\mathsf{H}}
    \right\}= \sigma^2 \boldsymbol{I}_{\tau_p}  \nonumber\\
&+ \mathbb{E} \big\{ \sum\nolimits_{k=1}^{K}\sqrt{p_k} \big[
 {{s}}^{(\tau_1)}_{t_k,n_1}J_{k,l,0}^{(\tau_1)},\cdots, {{s}}^{(\tau_p)}_{t_k,n_{\tau_p}}J_{k,l,0}^{(\tau_p)}
 \big]^{\mathsf{T}}
 {{h}}_{k,l,n}  \nonumber\\
&\,\,\,\,\,\,\,\, \times {{h}}_{k,l,n}^{*} \sqrt{p_k} \big[
 {{s}}^{(\tau_1)}_{t_k,n_1}J_{k,l,0}^{*(\tau_1)},\cdots, {{s}}^{(\tau_p)}_{t_k,n_{\tau_p}}J_{k,l,0}^{*(\tau_p)}
 \big]
 \big\}  \nonumber
 \\
 &+ \mathbb{E} \Big \{ \underbrace{\Big[\sum\nolimits_{k=1}^{K}\zeta_{k,l,n_1}^{(\tau_1)},\cdots, \sum\nolimits_{k=1}^{K}\zeta_{k,l,n_{\tau_p}}^{(\tau_p)}\Big]^{\mathsf{T}}}_{\triangleq \boldsymbol{\zeta}_{k,l}} \boldsymbol{\zeta}_{k,l}^{\mathsf{H}} 
 \Big\} \nonumber
 \nonumber\\
 &= \underbrace{\sum\nolimits_{k=1}^{K}p_{k}\beta_{k,l} \boldsymbol{\Phi}_{k} + \boldsymbol{Z}^{\text{ICI}}_{l} + \sigma^2 \boldsymbol{I}_{\tau_p}}_{\triangleq \boldsymbol{\Psi}_{l}}, 
\end{align}
where  the $(\tau_1,\tau_2)$ element of $ \boldsymbol{\Phi}_{k} \in \mathbb{C}^{\tau_p \times \tau_p}$ is 
\begin{align}
    [\boldsymbol{\Phi}_{k}]_{\tau_1, \tau_2} = {s}^{(\tau_1)}_{t_k,n_1} {s}^{*(\tau_2)}_{t_k,n_2} {B}_{0,0}^{(\tau_1-\tau_2)}.
\end{align} 
The $(\tau_1, \tau_2)$ element of the ICI component $\boldsymbol{Z}^{\text{ICI}}_{l} \in \mathbb{C}^{\tau_p \times \tau_p}$ is 
\begin{align}
    &[\boldsymbol{Z}^{\text{ICI}}_{l}]_{\tau_1, \tau_2}  \nonumber \\
    &= \mathbb{E}\Big\{\sum_{k=1}^K \zeta_{k,l,n_1}^{(\tau_1)} \sum_{k^{\prime}=1}^K \zeta_{k^{\prime},l,n_{2}}^{*(\tau_2)}\Big\} =
    \sum\limits_{k=1}^K \mathbb{E}\left\{ \zeta_{k,l,n_1}^{(\tau_1)} \zeta_{k^{},l,n_2}^{*(\tau_2)}\right\}\nonumber\\
    &=
    \sum_{k=1}^K p_k \beta_{k,l} \sum_{\substack{\jmath_1 \neq n_1}}^{N-1} \sum_{\substack{\jmath_2 \neq n_2}}^{N-1}\mathbb{E} \{ {{s}}^{(\tau_1)}_{t_k,\jmath_1} {{s}}^{*(\tau_2)}_{t_k,\jmath_2} J^{(\tau_1)}_{k,l,n_1-\jmath_1} J^{*(\tau_2)}_{k,l,n_2-\jmath_2} \} \nonumber\\
    &= 
    \sum_{k=1}^K p_k \beta_{k,l} \Big(\sum_{\substack{\jmath_1 \in \mathcal{N}_{p} \\ \jmath_1 \neq n_1}}^{N-1} \sum_{\substack{\jmath_2 \in \mathcal{N}_{p} \\ \jmath_2 \neq n_2}}^{N-1} {{s}}^{(\tau_1)}_{t_k,\jmath_1} {{s}}^{*(\tau_2)}_{t_k,\jmath_2}  {B}_{n_1-\jmath_1,n_2-\jmath_2}^{(\tau_1-\tau_2)} \nonumber\\
&\qquad\qquad
+ \sum\nolimits_{\substack{\jmath_1 \in \mathcal{N}_{d} \\ \jmath_1 \neq n_1}}^{N-1} \sum\nolimits_{\substack{\jmath_2 \in \mathcal{N}_{d} \\ \jmath_2 \neq n_2}}^{N-1}  {B}_{n_1-\jmath_1,n_2-\jmath_2}^{(\tau_1-\tau_2)}\Big)\label{eq:Zeta_ICI},
\end{align}
where the correlation term ${B}_{n_1-\jmath_1,n_2-\jmath_2}^{(\tau_1-\tau_2)}$ is calculated following \eqref{eq:B_matrix_compute}.

\section{Proof of Proposition 1}\label{sec:Appendix_SINR_lo_UatF}

Since the effective channels vary with each \ac{ofdm} symbol $\tau$, we follow the reference~\cite[Theorem 4.4]{demir2021foundations} using the \ac{uatf} bound to obtain an achievable SE for data subcarrier $n \in \mathcal{N}_d$ at  \ac{ofdm} symbol $\tau \in \{1 ,\cdots, \tau_c\}$. These achievable SEs are averaged over all $\tau_c$ \ac{ofdm} symbols to obtain~\eqref{eq:SE_UatF_PN_aware}. 

Specifically, by adding and subtracting the average effective channel $\mathbb{E}\left\{\boldsymbol{v}_{k,n}^{\mathsf{H},(\tau)} \boldsymbol{D}_k \boldsymbol{h}_{k,n}^{(\tau)}\right\}$, \eqref{eq:s_k_n_level4} is rewritten as
\begin{align}
    \hat{s}_{k,n}^{(\tau)} &= \mathbb{E}\left\{\boldsymbol{v}_{k,n}^{\mathsf{H},(\tau)} \boldsymbol{D}_k \boldsymbol{h}_{k,n}^{(\tau)}\right\}  s_{k,n}^{(\tau)}   + \nu_{k,n}^{(\tau)},
\end{align}
where the interference term is
\begin{align}
    &\nu_{k,n}^{(\tau)} = \Big(\boldsymbol{v}_{k,n}^{\mathsf{H},(\tau)} \boldsymbol{D}_k \boldsymbol{h}_{k,n}^{(\tau)}  - \mathbb{E}\left\{\boldsymbol{v}_{k,n}^{\mathsf{H},(\tau)} \boldsymbol{D}_k \boldsymbol{h}_{k,n}^{(\tau)}\right\}\Big)s_{k,n}^{(\tau)} +  \\& {\sum_{\substack{i=1 \\i\neq k}}^{K}\boldsymbol{v}_{k,n}^{\mathsf{H},(\tau)} \boldsymbol{D}_k \boldsymbol{h}_{i,n}^{(\tau)} s_{i,n}^{(\tau)}} 
    + {\sum_{i=1}^{K} \boldsymbol{v}_{k,n}^{\mathsf{H},(\tau)} \boldsymbol{D}_k \boldsymbol{\zeta}_{i,n}^{(\tau)}}
    + \boldsymbol{v}_{k,n}^{\mathsf{H},(\tau)} \boldsymbol{D}_k \boldsymbol{\eta}_{n}^{(\tau)}.\nonumber 
\end{align}
This can be viewed as a deterministic channel with a gain $\mathbb{E}\{\boldsymbol{v}_{k,n}^{\mathsf{H},(\tau)} \boldsymbol{D}_k \boldsymbol{h}_{k,n}^{(\tau)}\}$ and additive interference plus noise term $\nu_{k,n}^{(\tau)}$ that has zero mean. Note that $\nu_{k,n}^{(\tau)}$ is uncorrelated with the desired signal $s_{k,n}^{(\tau)}$ due to the independence  between each of the zero-mean symbols $s_{k,n}^{(\tau)}$, i.e., $\mathbb{E}\{s_{i,n}^{(\tau)} s_{j,n}^{(\tau)} \}=\mathbb{E}\{s_{i,n}^{(\tau)} s_{i,\jmath}^{(\tau)} \}=0$ for $n,j\in\mathcal{N}_d$.
The denominator of the SINR in~\eqref{eq:SINR_lo_UatF} is obtained by
\begin{align}
    \mathbb{E}&\{|\nu_{k,n}^{(\tau)}|^2\} 
    = \sum\nolimits_{i=1 }^{K} p_i  \mathbb{E}\Big\{\left|\boldsymbol{v}_{k,n}^{\mathsf{H},(\tau)} \boldsymbol{D}_k \boldsymbol{h}_{i,n}^{(\tau)}\right|^2\Big\}   +{{\rho}^{\text{ICI},(\tau)}_{k,n}}  \nonumber\\
    &\;\;\;  - p_k \left|\mathbb{E}\Big\{\boldsymbol{v}_{k,n}^{\mathsf{H}, (\tau)} \boldsymbol{D}_k \boldsymbol{h}_{k,n}^{(\tau)}\right\}\Big|^2 + \sigma^2\mathbb{E}\Big\{\left|\boldsymbol{v}_{k,n}^{(\tau)} \boldsymbol{D}_k \right|^2\Big\}.
\end{align}
Here, the ICI term ${\rho}^{\text{ICI},(\tau)}_{k,n}$ is computed as
\begin{align}
    \rho^{\text{ICI},(\tau)}_{k,n} \nonumber 
    &= \sum\nolimits_{i=1}^{K} \mathbb{E}\left\{ \boldsymbol{v}_{k,n}^{\mathsf{H},(\tau)} \boldsymbol{D}_k \boldsymbol{\zeta}_{i,n}^{(\tau)}    \boldsymbol{\zeta}_{i,n}^{\mathsf{H},(\tau)} \boldsymbol{D}_k^{} \boldsymbol{v}_{k,n}^{(\tau)} \right\} \nonumber\\
    &= \sum\nolimits_{i=1}^{K} \mathbb{E}\left\{ \boldsymbol{v}_{k,n}^{\mathsf{H},(\tau)} \boldsymbol{D}_k  \text{diag}(\boldsymbol{\lambda}_{i,n}^{(\tau)}) \boldsymbol{D}_k^{} \boldsymbol{v}_{k,n}^{(\tau)} \right\},\label{eq:rho_ICI}
\end{align}
where the $l$-th element of the ICI power $\boldsymbol{\lambda}_{i,n}^{(\tau)}\in \mathbb{C}^{L}$ is computed as
\begin{align}
{\lambda}_{i,n,l}^{(\tau)} &= \sum\nolimits_{\substack{ \jmath \neq n}}^{N-1} \mathbb{E}\{|s^{(\tau)}_{i,\jmath}|^2\} \mathbb{E}\{|{h}_{i,l,\jmath}|^2\}  \mathbb{E}\{|{{J}}_{i,l,n-\jmath}^{(\tau)}|^2\}   \nonumber \\
& = p_i \beta_{i,l}  \sum\nolimits_{\substack{ \jmath \neq n}}^{N-1} {B}_{n-j,n-j}^{(0)}
\nonumber\\
    &=p_i \beta_{i,l} (1-{B}_{0,0}^{(0)}),\label{eq:define_lamda}
\end{align}
where we view the unknown channels over subcarriers other than $n$ as random variables instead of realizations to reduce computation complexity.  In the ideal case of no \ac{pn}, ${\lambda}_{i,n,l}^{(\tau)}=0$ and $\rho_{k,n}^{\text{ICI},(\tau)}=0$, which turns the \ac{sinr} and \ac{se} expressions in~\eqref{eq:SINR_lo_UatF} and~\eqref{eq:SE_UatF_PN_aware} to be the same as in~\cite{bjornson2020scalable}. 

\balance

\bibliographystyle{IEEEtran}
\bibliography{reference_list}
\end{document}